\documentclass[twoside,leqno]{article}

\usepackage[letterpaper]{geometry}

\usepackage{hyperref}
\usepackage{ltexpprt}
\usepackage{prem}

\begin{document}

\newcommand\relatedversion{}

\title{\Large Shannon meets Gray:\\
Noise-robust, Low-sensitivity Codes with Applications in Differential Privacy\relatedversion}
\author{{David Rasmussen} Lolck\thanks{Basic Algorithms Research Copenhagen, University of Copenhagen}
\and Rasmus Pagh$^*$}

\date{}

\maketitle

\fancyfoot[R]{\scriptsize{Copyright \textcopyright\ 2024 by SIAM\\
Unauthorized reproduction of this article is prohibited}}

\begin{abstract} 
    \small\baselineskip=9pt Integer data is typically made differentially private by adding noise from a Discrete Laplace (or Discrete Gaussian) distribution.
We study the setting where differential privacy of a counting query is achieved using bit-wise randomized response, i.e., independent, random bit flips on the encoding of the query answer.

Binary \emph{error-correcting codes} transmitted through noisy channels with independent bit flips are well-studied in information theory.
However, such codes are unsuitable for differential privacy since they have (by design) high sensitivity, i.e., neighbouring integers have encodings with a large Hamming distance.
\emph{Gray codes} show that it is possible to create an efficient sensitivity 1 encoding, but are also not suitable for differential privacy due to lack of noise-robustness.

Our main result is that it is possible, with a constant rate code, to \emph{simultaneously} achieve the sensitivity of Gray codes \emph{and} the noise-robustness of error-correcting codes (down to the noise level required for differential privacy).
An application of this new encoding of the integers is an asymptotically faster, space-optimal differentially private data structure for histograms.

\end{abstract}
\section{Introduction}

Random noise in computing can both be a blessing and a curse.
In a nutshell, coding theory aims to \emph{amplify} the difference between different data sets such that even after adding random noise it is possible to recreate an original data set with high probability.
Conversely, differential privacy~\cite{DworkMNS06} deliberately adds noise to \emph{obscure} the difference between different data sets, such that ``neighboring'' data sets that differ only in the data of one individual become hard to distinguish.
In this paper we consider, for given integers $m$ and $d$, noise-robust binary encodings $\mathcal{C}_{\textup{enc}}(v)\in \{0,1\}^d$ for $v\in \{0,\dots,m-1\}$.
Our noise model is a \emph{binary symmetric channel}, meaning that each bit of $\mathcal{C}_{\textup{enc}}(v)$ is independently flipped with some probability $p$ upper bounded by a (sufficiently small) absolute constant.

In the differential privacy literature, reporting a noisy random bit is known as \emph{randomized response}~\cite{dwork2014algorithmic,warner65}.
It is known that such noisy encodings satisfy $\varepsilon$-differential privacy, where~$\varepsilon$ depends on $p$ and the \emph{sensitivity} of the encoding.
Unlike traditional uses of randomized response directly on the input data, we are interested in differential privacy in the context of counting problems where we want to estimate the number of data points that satisfy some predicate. 
Since two neighbouring datasets will have counts that differ by at most one, if we use an encoding $\mathcal{C}_{\textup{enc}}$ for the output of the counting problem, the sensitivity is the maximum Hamming distance between the encodings $\mathcal{C}_{\textup{enc}}(v)$ and $\mathcal{C}_{\textup{enc}}(v+1)$ for $v=0,\dots,m-2$.

Arguably, the symmetric binary channel is the simplest way in which one could possibly add noise in order to achieve differential privacy.
In comparison, adding (say) Laplace noise requires a relatively complex hardware/software system performing nontrivial arithmetic, making it considerably harder to verify and trust. 
Thus the question we address in this paper is: 

\medskip

\emph{Is it possible to achieve good efficiency, privacy and utility guarantees with a deterministic encoding of the integers passed through a binary symmetric channel?}

\medskip

\noindent
A positive answer to this question requires an ``error-correcting Gray code'' code with the following properties:
\begin{itemize}
    \item has short length (ideally close to $\log_2 m$), 
    \item has low sensitivity (ideally 1), and 
    \item is noise robust in the sense that from a noisy version of $\mathcal{C}_{\textup{enc}}(v)$ we can compute an estimate whose distribution is tightly concentrated around $v$.
\end{itemize}

\medskip

Figure~\ref{fig:overview} shows properties of four well-known types of integer encodings, all having at most two of these properties.
Our main result is that there exists an explicit and efficient code enjoying \emph{all} three properties.

\begin{figure}[t]
    \centering
    \begin{tabular}{|c|c|c|c|}
        \hline
        {\bf Encoding} & {\bf Length} & {\bf Sensitivity}  & {\bf Noise Robust} \\
        \hline
        Binary & $\log_2(m)$ & $\log_2(m)$ & No\\ 
        Gray & $\log_2(m)$ & $1$ & No\\ 
        Unary & $m$ & $1$ & Yes\\ 
        ECC & $O(\log m)$ & $\Omega(\log m)$ & Yes\\ 
        {\bf New} & $O(\log m)$ & $1$ & Yes\\ 
        \hline
    \end{tabular}
    \caption{Overview of properties of different encodings for integers in $\{0,\dots,m-1\}$. Noise robustness is relative to a noise model in which each bit is flipped with a fixed probability $p$, a sufficiently small positive constant. We require the encoded integer to be recovered with high probability, up to an additive noise term whose magnitude is bounded by a geometric distribution. All methods have encoding and decoding time that is polynomial in the length of the encoding.}
    \label{fig:overview}
\end{figure}

\subsection{Background}

{\bf Error-correcting codes.}
The study of error correction and communication through noisy channels was first introduced by Shannon~\cite{shannon-channel}. 
He showed that there exists capacity achieving codes, while never explicitly constructing them. 
Many explicit encoding schemes have since achieved a constant fraction of the capacity, including Reed-Solomon codes~\cite{reed-solomon-codes}, Justesen codes~\cite{justesen-codes} (the first such binary code), Expander codes~\cite{expander-codes,expander-codes-linear-time} (the first such linear time encodable/decodable code), Polar codes~\cite{polar-codes} (the first efficient code shown to achieve optimal capacity), and Reed-Muller codes (now known also to achieve optimal capacity~\cite{reeves2021reedmuller}).

{\bf Low-sensitivity codes.}
The \emph{reflected binary code} encodes the integers $\{0,\dots,m-1\}$ using $\lceil\log_2 m\rceil$ bits with the property that the encodings of consecutive integers differ only in one bit, i.e., the Hamming distance is~1.
Though codes with this property have been known at least since the 19th century, the reflected binary code is commonly referred to as the \emph{Gray code} after Frank Gray who described it in a 1947 patent application~\cite{knuth2011art}.

We are not aware of previous work that explicitly addresses combining error-correction capabilities with low sensitivity.
\emph{Locality properties} are important in several classes of error-correcting codes including locally decodable codes~\cite{Yekhanin11} and locally testable codes~\cite{dinur22,Panteleev22}, but this does not seem to translate into low sensitivity encodings of integers.
Another type of (non-binary) codes based on the Chinese Remainder Theorem~\cite{WangX10, Xiao20} have nontrivial bounds on sensitivity but do not seem to imply good binary codes.

{\bf Integer encodings in differential privacy.}
Aum{\"u}ller, Lebeda and Pagh~\cite{alp-mech} recently presented a differentially private mechanism, ALP, for representing a multiset $S$, where each data owner provides one of $n$ elements from a ground set $\{1,\dots,u\}$.
The output of the mechanism is a data structure that supports \emph{frequency queries}: Given an element $x$, it returns an estimate of the number of occurrences of $x$ in $S$.
The main difficulty of this problem lies in representing small frequencies, below a threshold $\ell = O(\log u)$, with good precision while keeping space close to the information-theoretical limit.
The idea of ALP is to represent frequencies up to $\ell$ in \emph{unary}, using hashing to determine the location of each bit.
To answer a frequency query for $x$ we inspect the bits at positions $h((x,1)),h((x,2)),\dots,h((x,\ell))$, where $h$ is a hash function.
Without privacy (i.e., $\varepsilon = \infty$), if the frequency is $f_x$ the bits at positions $h((x,1)),h((x,2)),\dots,h((x,f_x))$ are guaranteed to be $1$.
Some bits at positions $h((x,f_x+1)),\dots,h((x,\ell))$ may be 1 due to hash collisions, but we can determine $f_x$ from the data structure up to a small, geometrically distributed error.
To achieve pure differential privacy ALP uses \emph{randomized response}~\cite{warner65}, where each bit is flipped with probability depending on $\varepsilon$.
This works because the \emph{sensitivity} of a unary code is 1: Adding or removing an element changes at most one bit in the (non-private) data structure.
Answering a frequency query is done by taking the \emph{most likely} frequency, in a maximum likelihood sense, namely the one that is closest in Hamming distance to the sequence of observed bits.
This yields a tightly concentrated error distribution, comparable to using Laplace distributed noise to release each frequency $f_x$.
However, the use of a unary representation means that the number of bits $\ell$ needed to retrieve a frequency estimate is $\Theta(\log u)$.
We would prefer a more efficient encoding that can still tolerate the errors due to randomized response or hash collisions and has small sensitivity --- which is exactly what we achieve in this work.

\subsection{Our results}
Our main result shows how to transform an error-correcting code to a code that has sensitivity~1 and retains good error-correction properties.

\begin{theorem}[Informal version of \cref{thm:prob-diff-greater-than-t}]\label{thm:main}
    Let $\mathcal{C}$ be a code with block length $d$ and message length $\lg m$. Let $p\in (0;1/2)$ and let $b_p \sim \Bern(p)^{d'}$. Then there exists an explicit code $\mathcal{G}$ with block length $d' = O(d)$ and message length at least $\lg m$ consisting of the encoder $\mathcal{G}_{\textup{enc}}$ and decoder $\mathcal{G}_{\textup{dec}}$, were $\mathcal{G}$ has sensitivity $1$. Furthermore, for every $v\in [m]$ and $t>0$,
    \begin{equation*}
        \Pr[|v-\mathcal{G}_{\textup{dec}}(\mathcal{G}_{\textup{enc}}(v)\oplus b_p)| \ge t] \le e^{- \Omega(t)} + e^{-\Omega(d)} + O(P_p(\mathcal{C})),
    \end{equation*}
    where $P_p(\mathcal{C})$ is the probability of $\mathcal{C}$ being decoded incorrectly. If the encoding and decoding of $\mathcal{C}$ runs in polynomial time, then so does the encoding and decoding of $\mathcal{G}$.
\end{theorem}

An application of our codes is to replace the unary encodings used in the ALP mechanism of~\cite{alp-mech}.
For concreteness, we consider the case of representing histograms with $\varepsilon$-differential privacy (their Theorem~5.10 with $d=k=n$) where we can show:
\begin{theorem}\label{thm:alp}
    Given $\epsilon > 0$ and integers $n$, $u$, there exists a mechanism for releasing the histogram of a multiset of $n$ elements from $\{1,\dots,u\}$, producing an $\varepsilon$-differentially private data structure with the following properties:
    \begin{itemize}
        \item The space usage is $O(n\log u)$ bits,
        \item expected access time to a multiplicity estimate is $O(\log\log u)$,
        \item estimation error is $O(1/\varepsilon)$ in expectation and $O(\log(u)/\varepsilon)$ with high probability.
    \end{itemize}
\end{theorem}
This matches the privacy, space and error properties of~\cite{alp-mech}, which are optimal up to constant factors, while speeding up asymptotic access time exponentially.
Competing methods either use space proportional to $u$, have estimation error that is logarithmic in $u$, or query time that is proportional to $n$, see~\cite{alp-mech} for details.
The proof of Theorem~\ref{thm:alp} can be found in Appendix~\ref{sec:alp-proof}.

Since our code combines natural properties we believe it may have additional applications. For example, Acharya et al.~\cite{acharya2021distributed, acharya2023discrete} discuss how the lack of error robustness of the Gray code poses problems for some approaches to private, distributed mean estimation.

\subsection{Technical overview}

We provide two different reductions that transform a classical error-correcting code into an error-correcting Gray code $\mathcal{G}$.
The first reduction works for any error-correcting encoder-decoder pair $\mathcal{C}_{\textup{enc}}: [m] \rightarrow \{0,1\}^d$, $\mathcal{C}_{\textup{dec}}: \{0,1\}^d \rightarrow [m]$ and conceptually works in two steps. First, we construct a code $\mathcal{K}=(\mathcal{K}_{\textup{enc}}, \mathcal{K}_{\textup{dec}})$ with two key properties:
    \begin{itemize}
    \item Constant consecutive distance, meaning the Hamming distance between the encodings $\mathcal{K}_{\textup{enc}}(v)$, $\mathcal{K}_{\textup{enc}}(v+1)$ of consecutive integers is a fixed value $g$, independent of $v$, and
    \item decoding is possible even if, in addition to noise, the first or second half is replaced by bits from a different codeword.
    \end{itemize}
    We show that such a code with block length $d'=O(d)$ can be constructed by concatenating~4 copies of $\mathcal{C}_{\textup{enc}}(v)$, of which two are bit-wise negated if $v$ is odd, plus some padding bits that allow us to determine the parity of $v$.

The first step gives us codewords for every $v$ divisible by $g$ in which case we let $\mathcal{G}_{\textup{enc}}(v)=\mathcal{K}_{\textup{enc}}(v/g)$. 
    To obtain intermediate codewords we simply flip the $g$ bits that differ, one by one, such that the Hamming distance between consecutive codewords is 1.
    In this way, the codeword for $v$ is always a prefix of $\mathcal{K}_{\textup{enc}}(\lfloor v/g\rfloor)$ followed by a suffix of $\mathcal{K}_{\textup{enc}}(\lfloor v/g\rfloor + 1)$, meaning that we can recover either $\lfloor v/g\rfloor$ or $\lfloor v/g + 1\rfloor$ with high probability. 
    Finally, given $\lfloor v/g\rfloor$ or $\lfloor v/g \rfloor+ 1$ we can compute a maximum-likelihood estimate of $v \text{ mod } g$ that is tightly concentrated around the true value.

Our second reduction starts with a linear error-correcting code over $GF(2)$ given by its encoding matrix. As the encoding of a linear code is more predictable than a black-box code, this allows us to avoid working with constant consecutive distance codes, and results in an error-correcting Gray code with better constant factors.
Though our constructions can be instantiated with any code, we highlight the combination with expander codes and polar codes, respectively.

\section{Definitions}
For bitstrings $s,t$, we define the concatenation of $s$ and $t$ as $st$. We define the bitwise xor of $s$ and $t$ as $s\oplus t$, and the bitwise inverse of $s$ as $\overline{s}$. For a natural number $n$, we write $[n]$ to denote the set $\{0,1,\ldots,n-1 \}$. We denote $\Bern(p)$ as the Bernoulli distribution. For a vector $x$, we let $\norm{x}_1$ be the Manhattan norm, that is, the sum of all entries in the vector. In particular for binary vectors, this corresponds to the number of $1$ entries in the vector.   

We will only be exploring the setting of binary codes, that is, codes that uses a binary alphabet. This is especially meaningful for this work, since our motivation is to use the most simple type of noise to achieve differential privacy in the setting of counting.
\begin{Definition}[Code]
    A code $\mathcal{C} = (\mathcal{C}_{\textup{enc}},\mathcal{C}_{\textup{dec}})$ is a tuple of an encoding mapping $\mathcal{C}_{\textup{enc}}: [m]\to\{0,1\}^d$ and a decoding mapping $\mathcal{C}_{\textup{dec}}: \{0,1\}^d\to[m]$ such that $\mathcal{C}_{\textup{dec}}(\mathcal{C}_{\textup{enc}}(v)) = v$ for all $v\in[m]$. 
\end{Definition}
In this setting, we will refer to $d$ as the \textit{block length}. Often in literature, we also have the \textit{message length}, which is the number of bits needed to represent the value encoded. We will often be writing this as $\lg m$, signifying that we are able to encode $m$ different values, as we for our usage are more interested in $m$ than the number of bits needed to represent $m$.

We refer to the value $\mathcal{C}_{\textup{enc}}(v)$ as a codeword, signifying it is in the image of the code $\mathcal{C}$.
When looking at error-correcting codes there is often a need for measuring the efficiency. We will focus on the failure probabilities of codes when considering their performance.
\begin{Definition}[Failure Probability]
    We define the failure probability of a code $\mathcal{C} = (\mathcal{C}_{\textup{enc}},\mathcal{C}_{\textup{dec}})$ with block length $d$  and error $b_p\sim \Bern(p)^d$ as the probability,
    \begin{equation*}
        P_p(\mathcal{C}) = \max_{v\in [m]}\Pr[\mathcal{C}_{\textup{dec}}(\mathcal{C}_{\textup{enc}}(v)\oplus b_p) \neq v].
    \end{equation*}
\end{Definition}

What the failure probability encapsulates is the probability of a code failing under an adversarial choice of $v$ to be encoded. A second measure which can be used to reason about the effectiveness of codes is the distance of a code. This is informally a measure of how close the closest codewords of a code are to each other in terms of Hamming Distance.

\begin{Definition}[Hamming Distance]
    The Hamming distance $H : \{0,1\}^n\times\{0,1\}^n \to\mathbb{N}$ between two binary strings is defined as the number of bits where the two strings differ:
    \begin{equation*}
        H(a,b) = \left|\{i : a_i \neq b_i\}\right|.
    \end{equation*}
\end{Definition}
\begin{Definition}[Distance]
    The distance $D(\mathcal{C})$ of a code $\mathcal{C} = (\mathcal{C}_{\textup{enc}},\mathcal{C}_{\textup{dec}})$ with block length $d$ is defined as the minimum distance between any two distinct codewords:
    \begin{equation*}
        D(\mathcal{C}) := \min_{a,b\in[m], a\neq b} H(\mathcal{C}_{\textup{enc}}(a),\mathcal{C}_{\textup{enc}}(b))
    \end{equation*}
\end{Definition}
Finally, we will formally define what we are going to mean by a Gray code.
\begin{Definition}[Gray code]
    A Gray code is a code $\mathcal{G} = (\mathcal{G}_{\textup{enc}},\mathcal{G}_{\textup{dec}})$ with block length $d$ and message length $\lg m$ such that for each $v\in [m-1]$, $H(\mathcal{G}_{\textup{enc}}(v), \mathcal{G}_{\textup{enc}}(v + 1)) = 1.$
\end{Definition}
\section{Constructions}

In this section, we will look at how to construct an error-correcting Gray code. The construction will be done through reductions from a black-box error-correcting code. The goal is to prove \cref{thm:main}.

We will show two similar constructions that achieve the same bound, though with different constants. The first one will be constructed from a general code $\mathcal{C}$, where we make no assumptions about the structure of the code. This is shown in \cref{sec:con-from-gen}. 

In the second construction, we will show how the assumption that the code $\mathcal{C}$ is linear leads to using fewer repetitions, thereby lowering the constants associated with the code. This is shown in \cref{app:con-from-lin}.

Before we start with the actual codes, we are going to start by introducing unary codes. As mentioned in the introduction, these codes have most of the properties we are looking for, except they are very space-inefficient.

\begin{construction}[Unary Code]\label{con:unary-code}
    An unary code $\mathcal{U}  = (\mathcal{U}_{\textup{enc}},\mathcal{U}_{\textup{dec}})$ with block length $m$ and message length $\lg m$ is defined such that 
    \begin{align*}
        \mathcal{U}_{\textup{enc}}(v) &= 1^v 0^{m-v}\\
        \mathcal{U}_{\textup{dec}}(c) &= \argmin_{v\in [m]} H(\mathcal{U}_{\textup{enc}}(v),c).
    \end{align*}
\end{construction}
Note that there exist efficient decoding algorithms for the unary code, see e.g.~\cite{alp-mech}.

\subsection{Construction from a general code} \label{sec:con-from-gen}
In this section we will show how to construct an error-correcting Gray code from a general code. We will start with introducing what we call complement codes. The idea is that we want to transform a code $\mathcal{C}$ into a code $\mathcal{L}$ where even values are encoded unmodified, while odd values have their bits negated. This will lead to us being able to achieve constant distance between consecutive codewords in \cref{con:const-dist-code}.

\begin{construction}[Complement Code]\label{con:complement-code}
    Using a code $\mathcal{C}$ with block length $d$ and message length $\lg m$, we construct a code $\mathcal{L} = (\mathcal{L}_{\textup{enc}},\mathcal{L}_{\textup{dec}})$ called a \emph{complement code} with block length $d + D(\mathcal{C})$ and message length $m$. Define
    \begin{equation*}
        \mathcal{L}_{\textup{enc}}(v) = \begin{cases} \mathcal{C}_{\textup{enc}}(v)0^{D(\mathcal{C})} & \text{ if } $v$ \text{ is even}\\ \overline{\mathcal{C}_{\textup{enc}}(v)}1^{D(\mathcal{C})} & \text{ otherwise}\end{cases}
    \end{equation*}
    and for $c\in \{0,1\}^d, t\in \{0,1\}^{D(\mathcal{C})}$,
    \begin{equation*}
        \mathcal{L}_{\textup{dec}}(ct) = \begin{cases} \mathcal{C}_{\textup{dec}}(c) & \text{ if $t$ contains more $0$s than $1$s}\\ \mathcal{C}_{\textup{dec}}(\overline{c}) & \text{ if $t$ contains more $1$s than $0$s}\end{cases}
    \end{equation*}
    As a technical detail, if there is equally many $0$s and $1$s, then one of the options is chosen uniformly at random. It is clear that $\mathcal{L}$ is a code since for all $v\in [m]$, $\mathcal{L}_{\textup{dec}}(\mathcal{L}_{\textup{enc}}(v)) = v$ from the fact that $\mathcal{C}$ is a code.
\end{construction}

To start with, we are going to bound the error rate of a general code, which we can use to argue that the addition of the parity padding does not significantly worsen the quality of the code. The proof of this lemma can be found in \cref{app:code-fail-proof}.
\begin{lemma}\label{lem:code-fail-close-to-half-dist}
    Let $\mathcal{C} = (\mathcal{C}_{\textup{enc}},\mathcal{C}_{\textup{dec}})$ be a code with block length $d$ and message length $\lg m$, let $p\in (0, 1/2)$ and let $c_p \sim\Bern(p)^{D(C)}$. Then 
    \begin{equation*}
        P_p(\mathcal{C}) \ge \Pr[\norm{c_p}_1 > \frac{D(\mathcal{C})}{2}] + \frac{1}{2}\Pr[\norm{c_p}_1 = \frac{D(\mathcal{C})}{2}]
    \end{equation*}
\end{lemma}
What this lemma should be read as is, that any code would perform poorly if put in the same situation as the situations where the parity padding performs poorly. In the setting we are exploring we are mostly concerned with the failure probability so the next lemma lets us directly relate the failure probability of this construction to the failure probability of~$\mathcal{C}$. The lemma mostly argues and uses that any code must have an error probability that is lower bounded by that of decoding the wrong parity.

\begin{lemma}\label{lem:complement-code-error}
    Let $\mathcal{L}$ be a code constructed using \cref{con:complement-code} on the code $\mathcal{C}$. Then $P_p(\mathcal{L}) \le 2P_p(\mathcal{C})$.
\end{lemma}
\begin{proof}
    Let $m$ be the message length and let $d$ be the block length of $\mathcal{C}$. The code $\mathcal{L}$ can be decoded wrongly in two different ways. Either the parity bit-string is decoded incorrectly, or the inner code $\mathcal{C}$ is decoded incorrectly. 
    
    For the first case, observe that this is the setting of \cref{lem:code-fail-close-to-half-dist}, and so we can bound the probability of the parity code failing with probability $P_p(\mathcal{C})$. Taking a union bound with the probability of decoding the inner code incorrectly, we get $P_p(\mathcal{L}) \le 2P_p(\mathcal{C})$.
\end{proof}

We now show one of our main constructions. While we do insert a large amount of redundancy in the code through repetition, it is this repetition we will later use in the decoding process to rule out collisions of codewords.
\begin{construction}[Constant Consecutive Distance Code]\label{con:const-dist-code}
    Using the code $\mathcal{C}$ with block length $d$ and message length $\lg m$ and the code $\mathcal{L}$ that is obtained from \cref{con:complement-code} on $\mathcal{C}$, we construct a code $\mathcal{K} = (\mathcal{K}_{\textup{enc}},\mathcal{K}_{\textup{dec}})$ with block length $4d + 2D(\mathcal{C})$ and message length $\lg m$. We define 
    \begin{equation*}
        \mathcal{K}_{\textup{enc}}(v) := \mathcal{C}_{\textup{enc}}(v)\mathcal{L}_{\textup{enc}}(v)\mathcal{C}_{\textup{enc}}(v)\mathcal{L}_{\textup{enc}}(v).
    \end{equation*}
    
    Decoding is defined as follows: Let $c_1,c_2\in \{0,1\}^d, l_1,l_2\in \{0,1\}^{d + D(\mathcal{C})}$. Define the output of $\mathcal{K}_{\textup{dec}}(c_1l_1c_2l_2)$ as computing the four values $\mathcal{C}_{\textup{dec}}(c_1),\mathcal{C}_{\textup{dec}}(c_2),\mathcal{L}_{\textup{dec}}(l_1)$ and $\mathcal{L}_{\textup{dec}}(l_2)$ and outputting the most frequent one, breaking ties arbitrarily.
    
    Since $\mathcal{C}_{\textup{dec}}(\mathcal{C}_{\textup{enc}}(v))=\mathcal{L}_{\textup{dec}}(\mathcal{L}_{\textup{enc}}(v))=v$, decoding will also result in $v$, and so it holds that $\mathcal{K}_{\textup{dec}}(\mathcal{K}_{\textup{enc}}(v)) = v$ for all $v$, implying that $\mathcal{K}$ is a code.
\end{construction}
The following lemma gives one of the primary reasons to use this construction, namely that successive codewords differs by a fixed number of bits independent of $v$. This will make decoding feasible later. It should be mentioned that the number of bits still depends on $D(\mathcal{C})$.
\begin{lemma}\label{lem:const-dist}
    Let $\mathcal{K} = (\mathcal{K}_{\textup{enc}},\mathcal{K}_{\textup{dec}})$ be a code obtained from \cref{con:const-dist-code} on the code $\mathcal{C} = (\mathcal{C}_{\textup{enc}},\mathcal{C}_{\textup{dec}})$ with block length $d$ and message length $\lg m$. Then for any $v \in [m-1]$, 
    \begin{equation*}
        H(\mathcal{K}_{\textup{enc}}(v),\mathcal{K}_{\textup{enc}}(v + 1)) = 2 (d + D(\mathcal{C}))
    \end{equation*}
\end{lemma}
\begin{proof}
    From the construction of $\mathcal{K}$, we have
    \begin{equation}\label{eq:const-dist-hdist}
        H(\mathcal{K}_{\textup{enc}}(v),\mathcal{K}_{\textup{enc}}(v + 1)) = 2H(\mathcal{C}_{\textup{enc}}(v),\mathcal{C}_{\textup{enc}}(v + 1)) + 2H(\mathcal{L}_{\textup{enc}}(v),\mathcal{L}_{\textup{enc}}(v + 1)),
    \end{equation}
    where $\mathcal{C} = (\mathcal{C}_{\textup{enc}},\mathcal{C}_{\textup{dec}})$ and $\mathcal{L} = (\mathcal{L}_{\textup{enc}},\mathcal{L}_{\textup{dec}})$ is a code obtained from \cref{con:complement-code}. Since $v$ and $v + 1$ have different parity, we have that 
    \begin{equation*}
        H(\mathcal{L}_{\textup{enc}}(v),\mathcal{L}_{\textup{enc}}(v + 1)) = H(\mathcal{C}_{\textup{enc}}(v),\overline{\mathcal{C}_{\textup{enc}}(v + 1)}) + D(\mathcal{C}).
    \end{equation*}
    Observe that $H(\mathcal{C}_{\textup{enc}}(v),\overline{\mathcal{C}_{\textup{enc}}(v + 1)}) = d-H(\mathcal{C}_{\textup{enc}}(v),\mathcal{C}_{\textup{enc}}(v + 1))$, since $\mathcal{C}_{\textup{enc}}(v)$ differs from $\mathcal{C}_{\textup{enc}}(v + 1)$ in exactly the bit positions where $\mathcal{C}_{\textup{enc}}(v)$ is equal to $\overline{\mathcal{C}_{\textup{enc}}(v + 1)}$. This gives
    \begin{equation*}
        H(\mathcal{L}_{\textup{enc}}(v),\mathcal{L}_{\textup{enc}}(v + 1)) = d-H(\mathcal{C}_{\textup{enc}}(v),\mathcal{C}_{\textup{enc}}(v + 1)) + D(\mathcal{C}).
    \end{equation*}
    Substituting this into \eqref{eq:const-dist-hdist} completes the proof.
\end{proof}

The second important property of \cref{con:const-dist-code} is that because it is composed of $4$ codes, we can allow any one of the codes to be modified to a degree where we are unable to decode it correctly as we can discover the encoded value from the other three codes. As we might be unable to directly decide which codeword is the one we are unable to decode, the final decoding is decided by a majority vote.

For the next construction, we define the functions $\pre_i(s)$ on the binary string $s$ to be the prefix of $s$ containing $i$ characters, and similarly we define $\suf_i(s)$ to be the suffix of $s$ containing $i$ characters.

\begin{construction}[Error Correcting Gray Code]\label{con:error-correcting-gray-code}
    Let $\mathcal{C}$ be a code with block length $d$ and message length $\lg m$ and let $\mathcal{K} = (\mathcal{K}_{\textup{enc}},\mathcal{K}_{\textup{dec}})$ be a code obtained from \cref{con:const-dist-code} on $\mathcal{C}$. Let $g = 2 (d + D(\mathcal{C}))$. We define the Error Correcting Gray Code $\mathcal{G} = (\mathcal{G}_{\textup{enc}},\mathcal{G}_{\textup{dec}})$ with block length $4d + 2D(\mathcal{C})$ and message length $\lg mg$.

    Let $b^{(v)}_{1},...,b^{(v)}_{g}$ be the indices of the bits where $\mathcal{K}_{\textup{enc}}(v)$ and $\mathcal{K}_{\textup{enc}}(v + 1)$ are different, in sorted order and define $b^{(v)}_{0} = 0$. Observe that by \cref{lem:const-dist}, $\mathcal{K}_{\textup{enc}}(v)$ and $\mathcal{K}_{\textup{enc}}(v + 1)$ are different in exactly $g$ bits. 
    
    Let $v = qg + r$ where $0 \le r < g$. Then
    \begin{equation}\label{eq:error-correcting-gray-code-enc}
        \mathcal{G}_{\textup{enc}}(v) = \pre_{b^{(q)}_{r}}\left(\mathcal{K}_{\textup{enc}}(q + 1)\right)\suf_{g-b^{(q)}_{r}}\left(\mathcal{K}_{\textup{enc}}(q)\right).
    \end{equation}
    Define the decoding $\mathcal{G}_{\textup{dec}}(c)$ as follows: Let $h_v\in \{0,1\}^g$ such that 
    \begin{equation*}
        (h_{v})_i = \begin{cases}0 & \text{ if } c_{b^{(v)}_{i}} = \mathcal{K}_{\textup{enc}}(v)_{b^{(v)}_{i}}\\ 1 & \text{ if } c_{b^{(v)}_{i}} = \mathcal{K}_{\textup{enc}}(v + 1)_{b^{(v)}_{i}}\\\end{cases}.
    \end{equation*}
    From this construction, $h_v$ is an unary code. Let $\mathcal{U}_{\textup{dec}}$ be a decoder for unary codes and let $t = \mathcal{K}_{\textup{dec}}(c)$. We then end up with two candidate decodings, which we name 
    $$v_0~=~g(t-1) + \mathcal{U}_{\textup{dec}}(h_{t-1}) \text{ and } v_1 = gt + \mathcal{U}_{\textup{dec}}(h_t)$$
    
    Then we define
    \begin{equation}
        \mathcal{G}_{\textup{dec}}(c) = \argmin_{v \in \{v_0,v_1\}} H(c,\mathcal{G}_{\textup{enc}}(v))
    \end{equation}
    We refer to \cref{lem:g-gray-code} to show that this indeed is a Gray code. 
\end{construction}

\begin{figure}
    \centering
    \ctikzfig{codes}
    \caption{Sketch of the two different perspectives the code can be viewed. The upper one shows the four different component codes, while the lower one shows how it consists of a prefix and a suffix}
    \label{fig:code-sketch}
\end{figure}

\begin{Remark}
    Observe that when constructing $\mathcal{G}$ from $\mathcal{C}$ using \cref{con:error-correcting-gray-code}, the time complexity of the encoding and decoding of $\mathcal{G}$ is the same as that of $\mathcal{C}$, up to constant factors.
\end{Remark}

From the previous construction, we observe that we end up needing $4$ copies of the original code $\mathcal{C}$. We have to use at least $3$ to be able to uniquely decode the outer code, while we needed an even number to be able to easily decode the code from having constant distance.

In the next three lemmas, we are going to show that our construction indeed is a Gray code. Recall that a code has the property of $\mathcal{G}_{\textup{dec}}(\mathcal{G}_{\textup{enc}}(v)) = v$ for all $v$ and that a Gray code furthermore has sensitivity of $1$. We show that $\mathcal{G}$ has a sensitivity of $1$ in \cref{lem:one-ham-dist}. In \cref{lem:injective-code} we show that $\mathcal{G}_{\textup{enc}}(v)$ is injective. These two lemmas are combined in \cref{lem:g-gray-code} to show that $\mathcal{G}$ indeed is a Gray code.

\begin{lemma}\label{lem:one-ham-dist}
     Let $\mathcal{C}$ be a code with message length $\lg m$ and let $\mathcal{G} = (\mathcal{G}_{\textup{enc}},\mathcal{G}_{\textup{dec}})$ be obtained from \cref{con:error-correcting-gray-code} on $\mathcal{C}$. Then for any $v\in [m-1]$,
     \begin{equation*}
         H(\mathcal{G}_{\textup{enc}}(v),\mathcal{G}_{\textup{enc}}(v + 1)) = 1.
     \end{equation*}
\end{lemma}
\begin{proof}
    Let $\mathcal{C}$ have block length $d$ and let $g = 2(d + D(\mathcal{C}))$. Let $v = qg + r$ where $0 \le r < g$. From \cref{con:error-correcting-gray-code}, observe that if $r < g-1$, then $\mathcal{G}_{\textup{enc}}(v)$ and $\mathcal{G}_{\textup{enc}}(v + 1)$ are different only in bit $b^{(q)}_{r+1}$. Otherwise, if  $r = g-1$ then 
    \begin{align*}
        \mathcal{G}_{\textup{enc}}(v + 1) &= \pre_{b^{(q+1)}_{0}}\left(\mathcal{K}_{\textup{enc}}(q + 2)\right)\suf_{g-b^{(q + 1)}_{0}}\left(\mathcal{K}_{\textup{enc}}(q + 1)\right)\\
        &= \pre_{0}\left(\mathcal{K}_{\textup{enc}}(q + 2)\right)\suf_{g}\left(\mathcal{K}_{\textup{enc}}(q + 1)\right)\\
        &= \mathcal{K}_{\textup{enc}}(q + 1)\\
        &=  \pre_{b^{(q)}_{g}}\left(\mathcal{K}_{\textup{enc}}(q + 1)\right)\suf_{g-b^{(q)}_{g}}\left(\mathcal{K}_{\textup{enc}}(q)\right)
    \end{align*}
    Since $\mathcal{G}_{\textup{enc}}(v) = \pre_{b^{(q)}_{g-1}}\left(\mathcal{K}_{\textup{enc}}(q + 1)\right)\suf_{g-b^{(q)}_{g-1}}\left(\mathcal{K}_{\textup{enc}}(q)\right)$, this means that $\mathcal{G}_{\textup{enc}}(v)$ and $\mathcal{G}_{\textup{enc}}(v + 1)$ only differs in the bit $b^{(q)}_{g}$, proving the statement.
\end{proof}

\begin{lemma}\label{lem:injective-code}
    Let $\mathcal{C}$ be a code and let $\mathcal{G} = (\mathcal{G}_{\textup{enc}},\mathcal{G}_{\textup{dec}})$ be obtained using \cref{con:error-correcting-gray-code} on $\mathcal{C}$. Then $\mathcal{G}_{\textup{enc}}$ is injective.
\end{lemma}
\begin{proof}
    Let $\mathcal{G}$ have block length $d$ and message length $\lg m$. Let $v,v'\in [m]$ such that $\mathcal{G}_{\textup{enc}}(v) = \mathcal{G}_{\textup{enc}}(v') = w$. Let $g = 2 (d + D(\mathcal{C}))$. Assume without loss of generality that $v \le v'$. From the definition of $\mathcal{G}$, we can split $w$ into $4$ codewords $w = c_1l_1c_2l_2$, such that $c_1,c_2 \in \{0,1\}^d$ and $l_1,l_2 \in \{0,1\}^{d + D(\mathcal{C})}$. From \eqref{eq:error-correcting-gray-code-enc}, observe that $w = ps = p's'$ can be seen as composed of a prefix $p$ of $\mathcal{K}_{\textup{enc}}\left(\left\lfloor v/g\right\rfloor + 1\right)$ and a suffix $s$ of $\mathcal{K}_{\textup{enc}}\left(\left\lfloor v/g\right\rfloor\right)$. It can also be seen as a prefix $p'$ of $\mathcal{K}_{\textup{enc}}\left(\left\lfloor v'/g\right\rfloor + 1\right)$ and a suffix $s'$ of $\mathcal{K}_{\textup{enc}}\left(\left\lfloor v'/g\right\rfloor\right)$. This means that some codewords of $c_1,l_1,c_2,l_2$ are to the left of the split, at most one of them is split by the prefix and the suffix, and some of them are to the right of the split. There are therefore at least one codeword $x$ among $c_1,l_1,c_2,l_2$, for which $x$ is not split by neither $p$ and $s$ nor $p'$ and $s'$. Let $\mathcal{X}_{\textup{enc}}\in \{\mathcal{C}_{\textup{enc}},\mathcal{L}_{\textup{enc}}\}$ be the encoder used to encode $x$. It then holds that $x\in \left\{\mathcal{X}_{\textup{enc}}\left(\left\lfloor v/g\right\rfloor\right),\mathcal{X}_{\textup{enc}}\left(\left\lfloor v/g\right\rfloor + 1\right)\right\}$ and  $x\in \left\{\mathcal{X}_{\textup{enc}}\left(\left\lfloor v'/g\right\rfloor\right),\mathcal{X}_{\textup{enc}}\left(\left\lfloor v'/g\right\rfloor + 1\right)\right\}$. Since both $\mathcal{C}$ and $\mathcal{L}$ are codes, and $\mathcal{C}_{\textup{enc}}$ and $\mathcal{L}_{\textup{enc}}$ therefore injective, we can conclude that $\left|\left\lfloor v/g\right\rfloor - \left\lfloor v'/g\right\rfloor\right| \le 1$.
    
    Now, assume for the purpose of contradiction that $\left\lfloor v/g\right\rfloor + 1 = \left\lfloor v'/g\right\rfloor$. This means the codeword they share which is not split, $x$, must be fully contained in $p$ and in $s'$, implying that $s$ is fully contained in $s'$. See \cref{fig:prefixsuffix-sketch} for a sketch of this. We now consider any codeword $y$ fully contained in $s$. As $s$ is contained in $s'$ so is $y$. Since $y$ is fully contained in $s'$ then $y$ must be the result of encoding $\left\lfloor v'/g\right\rfloor$. This would however imply that $y$ is not in $s$ by injectivity of $\mathcal{C}_{\textup{enc}}$. Therefore such a $y$ cannot exist. Furthermore, no codeword $z$ can be split by $p$ and $s$. To see this, observe that $z$ would also be fully contained in $s'$, a contradiction since this implies $z$ is fully contained in $p$. We can therefore conclude that $s$ is empty. This is a contradiction since by construction, $s$ is non-empty. We conclude that $\left\lfloor v/g\right\rfloor = \left\lfloor v'/g\right\rfloor$.
    
    Next, for any $q\in [m/g-1]$, $H(\mathcal{G}_{\textup{enc}}(gq),\mathcal{G}_{\textup{enc}}(g(q + 1))) = g$ by \cref{lem:const-dist}, and at the same time for any $u\in [m-1]$, $H(\mathcal{G}_{\textup{enc}}(u),\mathcal{G}_{\textup{enc}}(v + 1)) = 1$. This means that for every $r\in [g]$, $H(\mathcal{G}_{\textup{enc}}(gq),\mathcal{G}_{\textup{enc}}(gq + r)) = r$ and so all encodings must be different. We conclude that $v = v'$, showing that $\mathcal{G}_{\textup{enc}}$ is injective.
\end{proof}
\begin{figure}
    \centering
    \ctikzfig{prefixsuffix}
    \caption{Sketch of the how $\mathcal{G}$ would have to look if the code was not injective.}
    \label{fig:prefixsuffix-sketch}
\end{figure}

\begin{lemma}\label{lem:g-gray-code}
    Let $\mathcal{C}$ be a code and let $\mathcal{G}$ be obtained using \cref{con:error-correcting-gray-code} on $\mathcal{C}$. Then $\mathcal{G}$ is a Gray code.
\end{lemma}
\begin{proof}
    Let $\mathcal{G}$ have block length $d$ and message length $\lg m$ and let $g = 2 (d + D(\mathcal{C}))$. To show that $\mathcal{G}$ indeed is a Gray code, we need to show that \cref{con:error-correcting-gray-code} is a code, or in other words that we decode it correctly. Let $v\in [m]$ and let $\mathcal{G}_{\textup{enc}}(v) = c_1l_1c_2l_2$ such that $c_1,c_2 \in \{0,1\}^d$ and $l_1,l_2 \in \{0,1\}^{d + D(\mathcal{C})}$. Let $g = 2 (d + D(\mathcal{C}))$. By construction, we have $s\not\in \left\{\floor*{v/g},\floor*{v/g + 1}\right\}$ for at most $1$ codeword $s\in\{\mathcal{C}_{\textup{dec}}(c_1),\mathcal{L}_{\textup{dec}}(l_1),\mathcal{C}_{\textup{dec}}(c_2),\mathcal{L}_{\textup{dec}}(l_2)\}$. By the pigeonhole principle, at least one of these options must therefore occur twice. This means that in the decoding, the most frequent element $t$ of the multiset $\{\mathcal{C}_{\textup{dec}}(c_1),\mathcal{L}_{\textup{dec}}(l_1),\mathcal{C}_{\textup{dec}}(c_2),\mathcal{L}_{\textup{dec}}(l_2)\}$ has the property that $t \in \left\{\floor*{v/g + 1}, \floor*{v/g}\right\}$. From the construction of the decoder, both these cases are considered.
    
    Next, we observe that using the prefix of $\mathcal{K}_{\textup{enc}}\left(\floor*{v/g} + 1\right)$ and the suffix of $\mathcal{K}_{\textup{enc}}\left(\floor*{v/g}\right)$ exactly corresponds to a unary encoding on the changing bits in \eqref{eq:error-correcting-gray-code-enc}, implying that decoding $h_{\floor{v/g}}$ with $\mathcal{U}_{\textup{dec}}$ determines the number of bits belonging to the prefix and the number belonging to the suffix, where $h_{t}$ is defined as in \cref{con:error-correcting-gray-code}. This means that $v\in \{v_0,v_1\}$ for $v_0 = g(t-1) + \mathcal{U}_{\textup{dec}}(h_{t-1})$ and $v_1 = gt + \mathcal{U}_{\textup{dec}}(h_t)$. Since $\mathcal{G}_{\textup{enc}}$ is injective by \cref{lem:injective-code}, $\mathcal{G}_{\textup{enc}}(v)\in \{\mathcal{G}_{\textup{enc}}(v_0), \mathcal{G}_{\textup{enc}}(v_1)\}$ and $\mathcal{G}_{\textup{enc}}(v_0)\neq \mathcal{G}_{\textup{enc}}(v_1)$. This means exactly one of $H(\mathcal{G}_{\textup{enc}}(v),\mathcal{G}_{\textup{enc}}(v_0))$ and $H(\mathcal{G}_{\textup{enc}}(v),\mathcal{G}_{\textup{enc}}(v_1))$ is equal to zero. As the decoder minimises the Hamming distance, this implies that $\mathcal{G}_{\textup{dec}}(\mathcal{G}_{\textup{enc}}(v)) = v$, showing that $\mathcal{G}$ is a code. It then follows from  \cref{lem:one-ham-dist} that $\mathcal{G}$ is a Gray code.
\end{proof}

As we now have established that \cref{con:error-correcting-gray-code} is indeed a code and that it has a sensitivity of $1$ we can now start looking at the error handling properties of the code. Since our code essentially is an unary code built on top of a black box error correction code we will start by looking at the probability that adding noise results in another decoding being obtained. We are able to bound this based on the Hamming distance between the two bitstrings.

\begin{lemma}\label{lem:calc-prob}
    Let $c_1,c_2\in \{0,1\}^d$ be bitstrings, and let $p\in [0,1/2)$ and let $b_p \sim \Bern(p)^d$. Then 
    \begin{equation*}
        \Pr[H(c_1\oplus b_p,c_2) \le H(c_1\oplus b_p,c_1)] \le \exp(- \frac{\left(1 - 2p\right)^{2}}{4 p + 2}H(c_1,c_2)).
    \end{equation*}
\end{lemma}
\begin{proof}
    Let $k= H(c_1,c_2)$. Note that if $H(c_1\oplus b_p,c_1) \ge H(c_1\oplus b_p,c_2)$, then at least $k/2$ of the $k$ bits where $c_1$ and $c_2$ are different must have been flipped. Letting $Y$ be the random variable denoting the number of the $k$ bits that have been flipped, we get
    \begin{equation*}
        \Pr[H(c_1\oplus b_p,c_2) \le H(c_1\oplus b_p,c_1)] \le \Pr[Y\ge k/2].
    \end{equation*}
    Due to independence, we can use a Chernoff bound. With $E[Y] = p k$ we get:
    \begin{equation*}
        \Pr[H(c_1\oplus b_p,c_2) \le H(c_1\oplus b_p,c_1)] 
        \le \Pr[Y\ge k/2] 
        = \Pr[Y\ge \frac{1}{2p} p k]
        \le \exp(- \frac{\left(1-2 p\right)^{2}}{4 p + 2}k),
    \end{equation*}
    completing the proof.
\end{proof}

The relation between any two values $v,v'\in [m]$ encoded with \cref{con:error-correcting-gray-code} depends in large part on $|v-v'|$. Recall that if $|v-v'|$ is large, then the component codes are going to be different. However if $|v-v'|$ is small then the decoding will be more like a unary code decoding. We show this formally with the next two lemmas. 

\begin{lemma}\label{lem:small-dist-close}
    Let $\mathcal{C} = (\mathcal{C}_{\textup{dec}},\mathcal{C}_{\textup{enc}})$ be a code with block length $d$ and message length $\lg m$, and let $\mathcal{G}=(\mathcal{G}_{\textup{dec}},\mathcal{G}_{\textup{enc}})$ be obtained using \cref{con:error-correcting-gray-code} on $\mathcal{C}$. Then for all $v,v'\in [m]$,  if $|v-v'| < D(\mathcal{C})$ then $H(\mathcal{G}_{\textup{enc}}(v),\mathcal{G}_{\textup{enc}}(v')) = |v-v'|$.
\end{lemma}
\begin{proof}
    Let $g=2(d+D(\mathcal{C}))$ and fix $v$. Assume without loss of generality that $v < v'$. Let $b^{(v)}_i$ be defined as in \cref{con:error-correcting-gray-code}. To show the statement it suffices to show that all indices $b^{(v)}_0,\ldots,b^{(v)}_g$ are unique, and that all indices $b^{(v)}_{g-D(\mathcal{C}) + 1},\ldots, b^{(v)}_{g}, b^{(v+1)}_{0},\ldots,b^{(v+1)}_{D(\mathcal{C})}$ are unique, since one of these is a superset of the bit indices that changes one by one when transforming $\mathcal{G}_{\textup{enc}}(v)$ to $\mathcal{G}_{\textup{enc}}(w)$ through the series $\mathcal{G}_{\textup{enc}}(v),\mathcal{G}_{\textup{enc}}(v+1),\ldots,\mathcal{G}_{\textup{enc}}(w)$ by the assumption $|v-v'| < D(\mathcal{C})$.
    
    Observe that, $b^{(v)}_0,\ldots,b^{(v)}_g$ are all unique by definition. Furthermore, by \cref{lem:const-dist} the first $g/2$ bit changes are found in the first two component codes of $\mathcal{G}$, while the last $g/2$ are found in the last two component codes. In other words, $b^{(v)}_{g-D(\mathcal{C})+1},\ldots, b^{(v)}_{g} > 2d + D(\mathcal{C})$, while $b^{(v+1)}_{0},\ldots,b^{(v+1)}_{D(\mathcal{C})-1} \le 2d+D(\mathcal{C})$. This means that there can be no duplicates between $b^{(v)}_{g-D(\mathcal{C})+1},\ldots, b^{(v)}_{g}$ and $b^{(v+1)}_{0},\ldots,b^{(v+1)}_{D(\mathcal{C})-1}$.
\end{proof}
\begin{lemma}\label{lem:large-dist-far}
    Let $\mathcal{C} = (\mathcal{C}_{\textup{dec}},\mathcal{C}_{\textup{enc}})$ be a code with block length $d$ and message length $\lg m$, and let $\mathcal{G}=(\mathcal{G}_{\textup{dec}},\mathcal{G}_{\textup{enc}})$ be obtained using \cref{con:error-correcting-gray-code} on $\mathcal{C}$. Then for all $v,v'\in [m]$, if $|v-v'| \ge D(\mathcal{C})$ then $H(\mathcal{G}_{\textup{enc}}(v),\mathcal{G}_{\textup{enc}}(v')) \ge D(\mathcal{C})$.
\end{lemma}
\begin{proof}
    Assume for the purpose of contradiction that there exist $v$ and $v'$ such that $|v-v'| \ge D(\mathcal{C})$ but $H(\mathcal{G}_{\textup{enc}}(v),\mathcal{G}_{\textup{enc}}(v')) < D(\mathcal{C})$. Let $\mathcal{G}_{\textup{enc}}(v) = c_1l_1c_2l_2$ and let $\mathcal{G}_{\textup{enc}}(v') = c_1'l_1'c_2'l_2'$ where $c_1,c_2,c_1',c_2' \in \{0,1\}^d$ and $l_1,l_2,l_1',l_2' \in \{0,1\}^{d + D(\mathcal{C})}$. Then $H(c_i,c_i') < D(\mathcal{C})$ and $H(l_i,l_i') < D(\mathcal{C})$ for $ \; i\in \{1,2\}$.

    Observe that by construction at most one of $c_1,l_1,c_2,l_2$ is not a codeword of their respective codes, and at most one of $c_1',l_1',c_2',l_2'$ is not a codeword of their respective codes. This means that there exist a pair: 
    $(x,y) \in \{(c_1,c_1'),(l_1,l_1'),(c_2,c_2'),(l_2,l_2')\}$
    such that both $x$ and $y$ are codewords of the same code. By the definition of $D(\mathcal{C})$, this implies that $x=y$, and so that $\left|\left\lfloor v/g\right\rfloor - \left\lfloor v'/g\right\rfloor\right| \le 1$ by injectivity of $\mathcal{C}_{\textup{enc}}$.

    Next, let $f_i$ be the index of the bit where $\mathcal{G}_{\textup{enc}}(v+i)$ and $\mathcal{G}_{\textup{enc}}(v+i+1)$ differ. Observe that since $x=y$, no value in $f_0,\ldots,f_{|v-v'|}$ can lie in the bit interval of $\mathcal{G}_{\textup{enc}}(v)$ and $\mathcal{G}_{\textup{enc}}(v')$ where $x$ and $y$ are placed respectively. This however directly implies that all $f_0,\ldots,f_{|v-v'|}$ are unique, since for any duplicates to exist, there would have to be at least one bit-flip in the interval covered by $x$ and $y$ in $\mathcal{G}_{\textup{enc}}(v)$ and $\mathcal{G}_{\textup{enc}}(v')$. This implies that $H(\mathcal{G}_{\textup{enc}}(v),\mathcal{G}_{\textup{enc}}(v'))=|v-v'|$, a contradiction.
\end{proof}
Finally, we will look at how well our the result of encoding and decoding is concentrated around the encoded value when adding noise after encoding. 
\begin{theorem}\label{thm:prob-diff-greater-than-t}
    Let $\mathcal{C} = (\mathcal{C}_{\textup{dec}},\mathcal{C}_{\textup{enc}})$ be a code with block length $d$ and message length $\lg m$, and let $\mathcal{G}=(\mathcal{G}_{\textup{dec}},\mathcal{G}_{\textup{enc}})$ be obtained using \cref{con:error-correcting-gray-code} on $\mathcal{C}$. Let $p\in(0,1)$ and let $b_p \sim \Bern(p)^{4d+2D(\mathcal{C})}$. Let $c=\left(1 - 2 p \right)^{2}/(4 p + 2)$. Then for all $t\ge 0$,
    \begin{equation*}
        \Pr[|v- \mathcal{G}_{\textup{dec}}(\mathcal{G}_{\textup{enc}}(v)\oplus b_p)| \ge t] \le \frac{2}{1-\exp(-c)} \exp(-ct) + 12d\exp(-cD(\mathcal{C})) + 5P_p(\mathcal{C})
    \end{equation*}
\end{theorem}
\begin{proof}
    Let $v' = \mathcal{G}_{\textup{dec}}(\mathcal{G}_{\textup{enc}}(v)\oplus b_p)$. Observe that $v'$ is a random variable. To show the statement we are going to split all possible decoding events for $v'$ into two sets, $S_1$ and  $S_2,$. $S_1$ contains all events such that $t \le |v-v'|< D(C)$. $S_2$ contains all events such that $|v-v'|\ge D(C)$. Observe that for any $v'$ such that $|v'-v| \ge t$, $v'\in S_1\cup S_2$. Finally, let $F$ be the event that at least one of the $3$ codewords of $\mathcal{G}_{\textup{enc}}(v) = c_1 l_1 c_2 l_2$ that are not a concatenation of two different codes, is decoded incorrectly.
    
    We can now rewrite the probability:
    \begin{align}
         \Pr[|v- v'| \ge t] &= \Pr[(|v-v'|\ge t) \cap F^c] + \Pr[|v-v'|\ge t |F]\Pr[F]\nonumber\\
         &\le \Pr[(S_1\cup S_2) \cap F^c] + \Pr[F]\nonumber\\
         &\le \Pr[S_1\cap F^c] + \Pr[S_2 \cap F^c] + \Pr[F].\label{eq:bounded-terms}
    \end{align}
    We will bound each of these terms.  

    We start by bounding $\Pr[S_1\cap F^c]$. By the definition of $F^c$, $v$ must have been considered in the second phase of the decoding. This means that $v'$ was chosen over $v$ implying that 
    \begin{equation*}
        H(\mathcal{G}_{\textup{enc}}(v)\oplus b_p,\mathcal{G}_{\textup{enc}}(v'))\le H(\mathcal{G}_{\textup{enc}}(v)\oplus b,\mathcal{G}_{\textup{enc}}(v))   
    \end{equation*}
    From this, we can bound the probability $\Pr[S_1\cap F^c]$ by the sum of probabilities of each possible value of $v'$ that is decodable in the event $S_1 \cap F^c$ being chosen over $v$. That is for all $w\in [m]$ such that $t \le |v-w| < D(\mathcal{C})$,
    \begin{align}
        \Pr[(v'=w)\cap(S_1\cap F^c)] &\le \Pr[H(\mathcal{G}_{\textup{enc}}(v)\oplus b_p,\mathcal{G}_{\textup{enc}}(w))\le H(\mathcal{G}_{\textup{enc}}(v)\oplus b_p,\mathcal{G}_{\textup{enc}}(v))]\nonumber\\
        &\le  \exp(-c H(\mathcal{G}_{\textup{enc}}(v),\mathcal{G}_{\textup{enc}}(w)))\nonumber\\
        &= \exp(-c |v-w|) \label{eq:s1-vvp-eq}
    \end{align}
    using \cref{lem:calc-prob,lem:small-dist-close}. Observe that for each value $l = |v-v'|$ there exists at most two possible values of $v'$. Summing over the different values of $l$, using \cref{eq:s1-vvp-eq}, and that it is a geometric progression we get
    \begin{align}
        \Pr[S_1\cap F^c] &= \sum_{l=t}^{D(\mathcal{C})-1} \Pr[(l=|v'-v|)\cap(S_1\cap F^c)]\nonumber\\
        &\le \sum_{l=t}^{D(\mathcal{C})-1} 2\exp(-c |v-w|)\nonumber\\
        &=2\frac{1 - \exp(-cD(\mathcal{C}))}{1-\exp(-c)} - 2\frac{1-\exp(-ct)}{1-\exp(-c)}\nonumber\\
        &\le \frac{2}{1-\exp(-c)} \exp(-ct) \label{eq:s1-fc-bound}
    \end{align}
    
    For the event $S_2\cap F^c$, observe only if $\left|\floor*{v/g}-\floor*{v'/g}\right|\le 1$ can both $v$ and $v'$ be considered during second decoding step. This means that at most $3g\le 12d$ different $v'$ values can be compared to $v$, the ones in the interval $[v-6d;v+6d]$. Of these, only the ones in the intervals $[v-6d;v-D(\mathcal{C})]$ and $[v+D(\mathcal{C});v+6d]$ are actually part of the event $S_2\cup F^c$ per definition. All other values of $v'$ have a probability of $0$ to be decoded by the fact that we are looking at events that are a subset of $F^c$. By \cref{lem:large-dist-far} we have $H(\mathcal{G}_{\textup{enc}}(v),\mathcal{G}_{\textup{enc}}(v')) \ge D(\mathcal{C})$. From \cref{lem:calc-prob}, we can therefore calculate the probabilities of these $v'$ values as:
    \begin{align}
        \Pr[S_2\cap F^c] \le&\sum_{w\in [m]} \Pr[(v'=w) \cap (S_2 \cap F^c)]\nonumber\\
        \le& \sum_{w = v-6d}^{v-D(\mathcal{C})} \Pr[(v'=w) \cap S_2 \cap F^c] + \sum_{w = v+D(\mathcal{C})}^{v+6d} \Pr[(v'=w) \cap S_2 \cap F^c]\nonumber\\
        \le& \sum_{w = v-6d}^{v-D(\mathcal{C})} \exp(- \frac{\left(1 - 2p\right)^{2}}{4 p + 2}H(\mathcal{G}_{\textup{enc}}(v),\mathcal{G}_{\textup{enc}}(w)))\nonumber + \sum_{w = v+D(\mathcal{C})}^{v+6d} \exp(- \frac{\left(1 - 2p\right)^{2}}{4 p + 2}H(\mathcal{G}_{\textup{enc}}(v),\mathcal{G}_{\textup{enc}}(w)))\nonumber\\
        \le& 12d\exp(-cD(\mathcal{C})). \label{eq:s2-fc-bound}
    \end{align}
    Finally, we can determine the probability of $F$ happening as the union bound of the probability of any of the $3$ component codes being decoded incorrectly, which means 
    \begin{equation}\label{eq:f-bound}
        \Pr[F] \le P_p(D(\mathcal{C})) + 2 P_p(D(\mathcal{L})) \le 5 P_p(D(\mathcal{C})),
    \end{equation}
    where $\mathcal{L}$ is a code obtained using \cref{con:complement-code} on $\mathcal{C}$ and the error probability is obtained from \cref{lem:complement-code-error}.
    Substituting \cref{eq:s1-fc-bound,eq:s2-fc-bound,eq:f-bound} into \cref{eq:bounded-terms}, we get
    \[
        \Pr[|v- v'| \ge t] \le \frac{2}{1-\exp(-c)} \exp(-ct) + 12d\exp(-cD(\mathcal{C})) + 5P_p(\mathcal{C}),
    \]
    as desired.
\end{proof}

\section{Concrete Error-Correcting Gray Codes}
In this section, we will be looking at instantiating the presented codes, and what kinds of guarantees these give us. The idea is that we will instantiate \cref{con:error-correcting-gray-code} with polar codes as well as expander codes to show the properties we are able to achieve with these codes.

First, we look at expander codes \cite{expander-codes,expander-codes-linear-time}. These codes are linear codes that cannot quite reach the information theoretical limit, but on the other hand, they are robust when the noise is lower than their decoding limit. Selecting~$\mathcal{C}$ to be an expander code constructed to be able to handle a ratio of $\alpha< 1/4$ bit flips we will look at how our code performs with an error of $p=\alpha/2$. We furthermore know that expander codes can be encoded and decoded in $O(d)$ time \cite{expander-codes-linear-time}.
\begin{lemma}\label{lem:expand-error}
    Let $\mathcal{C} = (\mathcal{C}_{\textup{enc}},\mathcal{C}_{\textup{dec}})$ be an expander code with block length $d$ that can correct all errors of at most $\alpha d$ bit flips, then 
    \begin{equation*}
        P_{\alpha/2}(\mathcal{C}) \le e^{-\alpha d/6}
    \end{equation*}
\end{lemma}
\begin{proof}
    Let $Y$ be the number of errors. Since we are guaranteed to be able to handle $\alpha d$ errors and the expected number of errors is $E[Y] = \alpha d/2$, we have
    we have
    \begin{equation*}
        P_{\alpha/2}(\mathcal{C}) \le \Pr[Y\ge \alpha d] \le \exp(-\alpha d/6)
    \end{equation*}
    using the Chernoff bound $\Pr[Y\ge 2E[Y]]\le \exp(-E[Y]/3)$.
\end{proof}
From this, we present the instantiation of our codes using expander codes,
\begin{corollary}\label{cor:expandercode}
    Let $\mathcal{C}$ be an expander code with block length $d = O_\alpha(\lg m)$ and message length $\lg m$ that can correct all errors of at most $\alpha d$ bit flips, where $\alpha < \frac{1}{4}$ and let $\mathcal{G} = (\mathcal{G}_{\textup{enc}},\mathcal{G}_{\textup{dec}})$ be a code constructed using \cref{con:error-correcting-gray-code} on $\mathcal{C}$.  Then $\mathcal{G}$ is a Gray code with message length of at least $\lg m$ and block length $d' \le 6d$ such that for all $t,v\in [m]$.
    \begin{equation*}
        \Pr[|v-\mathcal{G}_{\textup{dec}}(\mathcal{G}_{\textup{enc}}(v)\oplus b_{\alpha/2})| \ge t] \le  \exp(- \Omega(t)) + \exp(-\Omega(d))
    \end{equation*}
    where $b_{\alpha/2} \sim \Bern(\alpha/2)^{d'}$ with running time $O(d')$ for both encoding and decoding. In addition, we have $\mathcal{G}_{\textup{enc}}(0) = 0^{d'}$
\end{corollary}
\begin{proof}
    Let $c = (1-2p)^2/(4p+2) = (1-\alpha)^2/(2\alpha+2) > \frac{9}{40}$. By \Cref{thm:prob-diff-greater-than-t}, we have 
    \begin{equation*}
        \Pr[|v-v'| \ge t] \le \frac{2}{1-\exp(-c)} \exp(-ct) + 12d e^{-cD(\mathcal{C})} + 5P_p(\mathcal{C}).
    \end{equation*}
    From evaluation we find $\frac{2}{1-\exp(-c)} < 10$. Substituting the bound on $c$ and using \cref{lem:expand-error} yields:
    \begin{equation*}
         \Pr[|v-\mathcal{G}_{\textup{dec}}(\mathcal{G}_{\textup{enc}}(v)\oplus b_{\alpha/2})| \ge t] < 10 e^{- \frac{9}{40}t} + 12d e^{-\frac{9}{40}D(\mathcal{C})} + 5 e^{-\alpha d/6}
    \end{equation*}
    which can be simplified to the desired result since $D(\mathcal{C})\ge 2\alpha d$.

    To see that $\mathcal{G}_{\textup{enc}}(0) = 0^{d'}$, by construction we have $\mathcal{G}_{\textup{enc}}(0) = \mathcal{C}_{\textup{enc}}(0)\mathcal{C}_{\textup{enc}}(0)0^{D(\mathcal{C})}\mathcal{C}_{\textup{enc}}(0)\mathcal{C}_{\textup{enc}}(0)0^{D(\mathcal{C})}$. Then the fact that expander codes are linear implies $\mathcal{C}_{\textup{enc}}(0)=0^d$.
\end{proof}

Another family of codes that is worth considering for instantiation are polar codes \cite{polar-codes}. These codes are of interest since they achieve the capacity of the information in the channel. For a polar code $\mathcal{C}$ with message length $\lg m$ and block length $d$, the probability of error is $P_p(\mathcal{C}) = O(d^{-1/4})$, with a running time of the decoder and encoder of $O(d\lg d)$, see \cite{polar-codes}. This leads to the corollary:
\begin{corollary}
    Let $\mathcal{C}$ be a polar code with block length $d$ and message length $\lg m$ and let $\mathcal{G} = (\mathcal{G}_{\textup{enc}},\mathcal{G}_{\textup{dec}})$ be a code constructed using \cref{con:error-correcting-gray-code} on $\mathcal{C}$. Then $\mathcal{G}$ is a Gray code with message length of at least $\lg m$ and block length $d'\le 6d$ such that for all $t,v\in [m]$
    \begin{equation*}
        \Pr[|v-\mathcal{G}_{\textup{dec}}(\mathcal{G}_{\textup{enc}}(v)\oplus b_{p})| \ge t] \le e^{- \Omega(t)} +  O(d^{-1/4})
    \end{equation*}
    where $b_{p} \sim \Bern(p)^{d'}$ with running time $O(d\lg d)$ for both encoding and decoding.
\end{corollary}
\section{Acknowledgement}
We thank the reviewers for constructive and detailed feedback.
The authors are affiliated with Basic Algorithms Research Copenhagen (BARC), supported by the VILLUM Foundation grant 16582. Rasmus Pagh is supported by Providentia, a Data Science Distinguished Investigator grant from Novo Nordisk Fonden.
\bibliography{biblo}
\begin{appendices}
    \section{Construction from a linear code}\label{app:con-from-lin}
In this section, we show how to exploit the linear structure of codes to achieve better constants for the length of the code. The general idea is that we avoid having to use the constant consecutive distance code of \cref{con:const-dist-code}. 

To make notation simpler, we will be using the canonical binary code for integers. We will be writing it as $\mathcal{B} = (\mathcal{B}_{\textup{enc}},\mathcal{B}_{\textup{dec}})$. A property that is often achieved in the construction of codes is linearity. This property can be used to make some small optimisations to our construction.

\begin{Definition}
    A code $\mathcal{C} = (\mathcal{C}_{\textup{enc}},\mathcal{C}_{\textup{dec}})$ where $\mathcal{C}_{\textup{enc}}: [2^n]\to GF(2)^d$ and a decoding mapping $\mathcal{C}_{\textup{dec}}: GF(2)^d\to[2^n]$ is called linear if there exists a $n\times d$ generator matrix $G\in GF(2)^{n\times d}$ such that $\mathcal{C}_{\textup{enc}}(v) = \mathcal{B}_{\textup{enc}}(v)^\intercal G$ for all $v\in [2^n]$.
\end{Definition}

Previously, we used the constant consecutive distance code to encode and decode efficiently. However, it also meant that we ended up needing an even number of repetitions of the code $\mathcal{C}$. It is however not enough to only use $2$ for a black-box error correcting code, so we required $4$ repetitions of the underlying $\mathcal{C}$. In addition to this, we also had to use some additional padding to communicate the parity of the encoded integer. In this section, we show how we can achieve the same properties, using only $3$ repetitions of $\mathcal{C}$ and no additional padding, while only getting an 

The following lemma is central to our ability to accomplish this.
\begin{algorithm}
\caption{An algorithm for computing $\sum_{i=1}^m H(\mathcal{C}_{\textup{enc}}(i-1),\mathcal{C}_{\textup{enc}}(i))$}\label{alg:total-ham-distance}
    \begin{algorithmic}
        \Function{CountCodeWords}{$t \in \mathbb{N}$, $G$ : $n\times d$ matrix such that $\mathcal{C}_{\textup{enc}}(v) = \mathcal{B}(v)^\intercal G$}
            \Let{$v_1$}{$\mathbf{0}$}
            \Let{$s$}{$0$}
            \For{$i = 1 \ldots n$}
                \Let{$v_i$}{$v_{i-1} \oplus G_{i\cdot}$}\Comment{$G_{i\cdot}$ denotes the $i$th row of $G$}
                \Let{$s$}{$s + \norm{v_i}_1\cdot \floor*{\frac{t}{2^i} + \frac{1}{2}}$}
            \EndFor
            \Return $s$
        \EndFunction
\end{algorithmic}
\end{algorithm}

\begin{lemma}
    Let $\mathcal{C}$ be a linear code with block length $d$ and message length $n$ and $n \times d$ generator matrix $G$. Then \cref{alg:total-ham-distance} computes 
    \begin{equation*}
        \textsc{CountCodeWords}(t,G) = \sum_{i=1}^t H(\mathcal{C}_{\textup{enc}}(i-1),\mathcal{C}_{\textup{enc}}(i))
    \end{equation*}
    in $O(nd)$ time.
\end{lemma}
\begin{proof}
    The time complexity is clear from the fact that $v_i$ takes $O(d)$ time to compute and that this is done for $i=1,\ldots,n$.
    
    For correctness, observe that for all $i$, $\mathcal{B}_{\textup{enc}}(i-1)\oplus\mathcal{B}_{\textup{enc}}(i) = 0^{n-k}1^{k}$ for some $k$.
    Since $\mathcal{C}$ is linear, by definition $\mathcal{C}_{\textup{enc}}(v) = \mathcal{B}(v)^\intercal G$ for some $n\times d$ matrix $G$. Letting $G_{i\cdot}$ denote the $i$th row of $G$ and $\mathcal{B}(v)_i$ the $i$th entry of $\mathcal{B}(v)$, we observe that 
    \begin{equation*}
        \mathcal{C}_{\textup{enc}}(v) = \bigoplus_{i = 1}^{n} (\mathcal{B}(v)_i\cdot G_{i\cdot})
    \end{equation*}
    Since $\mathcal{B}_{\textup{enc}}(i-1)$ and $\mathcal{B}_{\textup{enc}}(i)$ only differs on the last $k$ bits, this implies that
    \begin{equation}\label{eq:adj-ham-dist-equation}
        H(\mathcal{C}_{\textup{enc}}(i-1),\mathcal{C}_{\textup{enc}}(i)) = \norm{\bigoplus_{i = 1}^{k} G_{i\cdot}}_1 = \norm{v_k}_1
    \end{equation}
    
    It now simply remains for each $k$ value to count the number of $i$'s such that $\mathcal{B}_{\textup{enc}}(i-1)$ and $\mathcal{B}_{\textup{enc}}(i)$ differs exactly on the last $k$ bits. As the $i$th bit flips every $2^{i-1}$ increases by $1$, we observe that it flips to $1$ every $2^i$ increases. As it starts at $0$, we end up having the number of $i$'s that differ on exactly the $k$ last bits be $\floor*{\frac{m}{2^i} + \frac{1}{2}}$. Combining this with \cref{eq:adj-ham-dist-equation} we get
    \begin{equation*}
        \sum_{i=1}^m H(\mathcal{C}_{\textup{enc}}(i-1),\mathcal{C}_{\textup{enc}}(i)) = \sum_{k=1}^n \norm{v_i}_1\cdot \floor*{\frac{m}{2^i} + \frac{1}{2}}
    \end{equation*}
    which is exactly the value computed by \cref{alg:total-ham-distance}.
\end{proof}

\begin{construction}\label{con:repeat-code}
    Let $\mathcal{C}$ be a code with block length $d$ and message length $m$. We construct the code $\mathcal{W} = (\mathcal{W}_{\textup{enc}},\mathcal{W}_{\textup{dec}})$ with block length $3d$ and message length of $\lg m$.
    Define 
    \begin{equation*}
        \mathcal{W}_{\textup{enc}}(v) = \mathcal{C}_{\textup{enc}}(v)\mathcal{C}_{\textup{enc}}(v)\mathcal{C}_{\textup{enc}}(v)
    \end{equation*}
    \begin{equation*}
        \mathcal{W}_{\textup{dec}}(c_1 c_2 c_3) = \text{ Median of } \{\mathcal{C}_{\textup{dec}}(c_1), \mathcal{C}_{\textup{dec}}(c_2),\mathcal{C}_{\textup{dec}}(c_3)\}
    \end{equation*}
\end{construction}
Notice that we specifically use the median instead of a majority vote for decoding. The reason is we want to be able to decode the code, even if one of the codewords has been modified to the point where it can not be decoded. However, from our later construction, we also cannot guarantee that the two non-broken codewords encode the same value, just that the values are numerically adjacent. By selecting the median, we guarantee that one of the two non-broken codewords is the one returned.

With the constructions we now have, we can construct the linear code-based error-correcting Gray code. Note that though this code is based upon a linear code, we do not make any claims that the code is linear and in general it will not be linear.
\begin{construction}[Linear Code Based Error Correcting Gray Code]\label{con:linear-error-correcting-gray-code}
    Let $\mathcal{C} = (\mathcal{C}_{\textup{enc}},\mathcal{C}_{\textup{dec}})$ be a linear code with block length $d$ and message length $\lg m$ and let $\mathcal{W}$ be a code constructed using \cref{con:repeat-code} on~$\mathcal{C}$. We will construct the code $\mathcal{G} = (\mathcal{G}_{\textup{enc}},\mathcal{G}_{\textup{dec}})$ with block length $3d$ and message length at least $\lg m$.
    
    Let $s_v = H(\mathcal{W}_{\textup{enc}}(v),\mathcal{W}_{\textup{enc}}(v + 1))$. Let $b^{(v)}_{1},...,b^{(v)}_{s_v}$ be the bit indices where $\mathcal{W}_{\textup{enc}}(v)$ and $\mathcal{W}_{\textup{enc}}(v + 1)$ are different in sorted order and define $b^{(v)}_{0} = 0$. 
    Let $v = q + r$ such that 
    \begin{equation*}
        q = \sum_{i=1}^l H(\mathcal{C}_{\textup{enc}}(i-1),\mathcal{C}_{\textup{enc}}(i)) \le v < \sum_{i=1}^{l + 1} H(\mathcal{C}_{\textup{enc}}(i-1),\mathcal{C}_{\textup{enc}}(i))
    \end{equation*}
    for some $l$ and $0 \le r < H(\mathcal{C}_{\textup{enc}}(l),\mathcal{C}_{\textup{enc}}(l + 1))$. Then
    \begin{equation*}
        \mathcal{G}_{\textup{enc}}(v) = \pre_{b^{(l)}_{r}}\left(\mathcal{W}_{\textup{enc}}(l + 1)\right)\suf_{s_l-b^{(l)}_{r}}\left(\mathcal{W}_{\textup{enc}}(l)\right).
    \end{equation*}
    We define the decoding function in the following way. For input $c$, let $t = \mathcal{W}_{\textup{dec}}(c)$.  Then define the the bitstring $h_v\in \{0,1\}^{s_v}$ such that 
    \begin{equation*}
        (h_{v})_i = \begin{cases}0 & \text{ if } c_{b^{(v)}_{i}} = \mathcal{K}_{\textup{enc}}(v)_i\\ 1 & \text{ if } c_{b^{(v)}_{i}} = \mathcal{K}_{\textup{enc}}(v + 1)_i\\\end{cases}.
    \end{equation*}
    From this construction, $h_v$ becomes a unary code for all $v$. Letting $\mathcal{U}_{\textup{dec}}$ be the function for decoding unary functions. We then end up with two alternative decodings, which we name 
    \begin{equation*}
        v_0 = \mathcal{U}_{\textup{dec}}(h_{t-1}) + \sum_{i=1}^{t-1} H(\mathcal{C}_{\textup{enc}}(i-1),\mathcal{C}_{\textup{enc}}(i)),
    \end{equation*}
    \begin{equation*}
        v_1 =  \mathcal{U}_{\textup{dec}}(h_t) + \sum_{i=1}^{t} H(\mathcal{C}_{\textup{enc}}(i-1),\mathcal{C}_{\textup{enc}}(i)).
    \end{equation*}
    
    We define $\mathcal{G}_{\textup{dec}}(c) = \argmin_{v \in \{v_0,v_1\}} H(c,\mathcal{G}_{\textup{enc}}(v))$. We refer to \cref{lem:g-gray-code-lin} to show that this actually is a Gray code. 
\end{construction}

It is straightforward to see that $q$ can be efficiently computed in time $O(n^2d)$ through a binary search on \cref{alg:total-ham-distance}.

\begin{lemma}\label{lem:one-ham-dist-linear}
     Let $\mathcal{C}$ be a code with block length $d$ and message length $\lg m$ and let $\mathcal{G} = (\mathcal{G}_{\textup{enc}},\mathcal{G}_{\textup{dec}})$ be obtained using \cref{con:linear-error-correcting-gray-code} on $\mathcal{C}$. Then for any $v\in [m-1]$ 
     \begin{equation*}
         H(\mathcal{G}_{\textup{enc}}(v),\mathcal{G}_{\textup{enc}}(v + 1)) = 1
     \end{equation*}
\end{lemma}
\begin{proof}(Sketch)
    Use the same approach as in \cref{lem:one-ham-dist}, but use the definitions of $q$, $r$ and $\mathcal{G}_{\textup{enc}}(v)$ as they are in \cref{con:linear-error-correcting-gray-code}.
\end{proof}

\begin{lemma}\label{lem:injective-code-lin}
    Let $\mathcal{C}$ be a code and let $\mathcal{G} = (\mathcal{G}_{\textup{enc}},\mathcal{G}_{\textup{dec}})$ be constructed using \cref{con:linear-error-correcting-gray-code} on $\mathcal{C}$. Then $\mathcal{G}_{\textup{enc}}$ is injective.
\end{lemma}
\begin{proof}(Sketch)
    The approach is exactly the same as in \cref{lem:injective-code}, except that there are only three codewords instead of four. Instead of using $\left\lfloor v/g\right\rfloor$ to determine $q$, use the same method as in \cref{con:linear-error-correcting-gray-code}.
\end{proof}
\begin{lemma}\label{lem:g-gray-code-lin}
    Let $\mathcal{C}$ be a code with block length $d$ and message length $\lg m$ and let $\mathcal{G}$ be constructed using \cref{con:linear-error-correcting-gray-code} on $\mathcal{C}$. Then $\mathcal{G}$ is a Gray code.
\end{lemma}
\begin{proof}(Sketch)
    The structure of the proof is essentially the same as for \cref{lem:g-gray-code}. To show that \cref{con:linear-error-correcting-gray-code} is a Gray code, we start by showing that it is a code. Let $v\in [m]$ and let $l$ be the unique integer such that 
    \begin{equation*}
        \sum_{i=1}^l H(\mathcal{C}_{\textup{enc}}(i-1),\mathcal{C}_{\textup{enc}}(i)) \le v < \sum_{i=1}^{l + 1} H(\mathcal{C}_{\textup{enc}}(i-1),\mathcal{C}_{\textup{enc}}(i))
    \end{equation*}
    and let $\mathcal{G}_{\textup{enc}}(v) = c_1c_2c_3$ such that $c_1,c_2,c_3 \in \{0,1\}^d$. From construction, it holds that for at most $1$ codeword $s$ of $\{\mathcal{C}_{\textup{dec}}(c_1),\mathcal{C}_{\textup{dec}}(c_2),\mathcal{C}_{\textup{dec}}(c_3)\}$ that $s\not\in\{l,l + 1\}$. Regardless of the value of $s$, the median of $\mathcal{C}_{\textup{dec}}(c_1),\mathcal{C}_{\textup{dec}}(c_2)$ and $\mathcal{C}_{\textup{dec}}(c_3)$ is either $l$ or $l+1$. From here, the approach is the same as for \cref{lem:g-gray-code}.
\end{proof}
\begin{theorem}
    Let $\mathcal{C} = (\mathcal{C}_{\textup{dec}},\mathcal{C}_{\textup{enc}})$ be a code with block length $d$ and message length $m$, and let $\mathcal{G}=(\mathcal{G}_{\textup{dec}},\mathcal{G}_{\textup{enc}})$ be obtained using \cref{con:linear-error-correcting-gray-code} on $\mathcal{C}$. Let $p\in(0,1)$ and let $b_p \sim \Bern(p)^{3d}$. Let $c=\left(1 - 2 p \right)^{2}/(4 p + 2)$. Then for all $t\ge 0$,
    \begin{equation*}
        \Pr[|v-v'| \ge t] \le \frac{2}{1-\exp(-c)} \exp(- ct) + 9d \exp(-cD(\mathcal{C})) + 2P_p(\mathcal{C})
    \end{equation*}
\end{theorem}
\begin{proof}(Sketch)
    The proof of this is mostly equivalent to the proof of \cref{thm:prob-diff-greater-than-t}, except that \cref{con:linear-error-correcting-gray-code} only consists of $3$ copies of $\mathcal{C}$ and there are at most $9d$ values of $v'$ which can be considered together with $v$.
\end{proof}

    \section{Proof of \cref{lem:code-fail-close-to-half-dist}}\label{app:code-fail-proof}
In this section we prove \cref{lem:code-fail-close-to-half-dist}. The general idea is to show that for two chosen codewords, at least one of them must be decoded incorrectly with some bounded probability.
\begin{proof}%
    Pick $v,w \in [m]$ such that $H(\mathcal{C}_{enc}(v),\mathcal{C}_{enc}(w)) = D(\mathcal{C})$ and let $b_p\sim \Bern(p)^d$. Observe that since $p < \frac{1}{2}$, 
    \begin{equation}\label{eq:cp-less-less-than-greater}
        \Pr[\norm{c_p}_1 > \frac{D(\mathcal{C})}{2}] \le \Pr[\norm{c_p}_1 < \frac{D(\mathcal{C})}{2}]
    \end{equation}
    Let $c_p\in \{0,1\}^{D(\mathcal{C})}$ be the random bitstring where each bit in $c_p$ corresponds to a unique bit in $b_p$ where $\mathcal{C}_{enc}(v)$ and $\mathcal{C}_{enc}(w)$ are different. Let $(c_p)_i=1$ if and only if the corresponding bit in $b_p$ is $1$. Observe that $c_p \sim \Bern(p)^{D(C)}$. This means that $\norm{c_p}_1$ becomes a count over how many of the bits where $\mathcal{C}_{enc}(v)$ and $\mathcal{C}_{enc}(w)$ are different have been flipped.
    
    For the sake of notation, let $I(x)$ be the event $\mathcal{C}_{dec}(\mathcal{C}_{enc}(x) \oplus b_p) \neq x$ i.e., the event that the decoding fails. Next, observe that if $x = \mathcal{C}_{enc}(v)\oplus b_p$, and $x$ differs from $\mathcal{C}_{enc}(v)$ in $r$ of the positions where $\mathcal{C}_{enc}(v)$ and $\mathcal{C}_{enc}(w)$ differ, then it must hold that $x$ differs from $\mathcal{C}_{enc}(w)$ in exactly $D(\mathcal{C}) - r$ of these positions. This means that we can bound the probability as
    \begin{equation}\label{eq:bound-small-large}
        \Pr[I(v) \left| \norm{c_p}_1 < \frac{D(\mathcal{C})}{2}\right.] +\Pr[I(w) \left| \norm{c_p}_1 > \frac{D(\mathcal{C})}{2}\right.]\ge 1,
    \end{equation}
    and the same symmetrically
    \begin{equation}\label{eq:bound-equal-equal}
        \Pr[I(v) \left| \norm{c_p}_1 > \frac{D(\mathcal{C})}{2}\right.] + \Pr[I(w) \left| \norm{c_p}_1 < \frac{D(\mathcal{C})}{2}\right.]\ge 1,
    \end{equation}
    and finally
    \begin{equation}\label{eq:bound-large-small}
        \Pr[I(v) \left| \norm{c_p}_1 = \frac{D(\mathcal{C})}{2}\right. ] + \Pr[I(w) \left| \norm{c_p}_1 = \frac{D(\mathcal{C})}{2}\right.]\ge 1,
    \end{equation}
    since the conditional restrict to the same set of events and the decoding means the complements are disjoint.
    
    Now let $S$ the event that $\norm{c_p} < \frac{D(\mathcal{C})}{2}$, $T$ the event that $\norm{c_p} = \frac{D(\mathcal{C})}{2}$, and $U$ the event that $\norm{c_p} > \frac{D(\mathcal{C})}{2}$. This means we can write 
    \begin{equation*}
        \Pr[I(v)] = \Pr[I(v)|S]\Pr[S] + \Pr[I(v)|T]\Pr[T]+\Pr[I(v)|U]\Pr[U]
    \end{equation*}
    \begin{equation*}
        \Pr[I(w)] = \Pr[I(w)|S]\Pr[S] + \Pr[I(w)|T]\Pr[T]+\Pr[I(w)|U]\Pr[U]
    \end{equation*}
    Using \cref{eq:cp-less-less-than-greater,eq:bound-small-large,eq:bound-equal-equal,eq:bound-large-small}, we get
    \begin{align*}
        \Pr[I(v)] + \Pr[I(w)] \ge& (\Pr[I(v)|S] + \Pr[I(w)|U])\Pr[U] + (\Pr[I(v)|T] + \Pr[I(w)|T])\Pr[T]\\&+(\Pr[I(v)|U] + \Pr[I(w)|S])\Pr[U]\\
        \ge& \Pr[T] + 2\Pr[U].
    \end{align*}
    Since at least one of $\Pr[I(v)]$ and $\Pr[I(w)]$ must be greater than the average, this proves the lemma.
\end{proof}
    \section{Proof of Theorem~\ref{thm:alp}}\label{sec:alp-proof}
We follow the approach of Aum{\"u}ller, Lebeda and Pagh~\cite{alp-mech} for the case of pure differential privacy.
This section is intended to be read alongside parts of their paper that we refer to and modify.\footnote{We note that the variables $m$ and~$\alpha$ used in the present paper to denote parameters of codes are used for other quantities in~\cite{alp-mech}. Conversely, in this section, we use $b$ consistently with~\cite{alp-mech} to denote an index in the data structure, while $b$ is used in previous sections to denote an error vector.}
For a multiset $S$ we denote the frequency of element $i\in S$ by $x_i$ and let $x_i = 0$ for $i\in [u]\backslash S$, such that $S$ is encoded by $x\in \{0,1\}^u$ which is $n$-sparse.
We modify the ALP1-Projection algorithm (Algorithm 2 in~\cite{alp-mech}) by replacing the choice of $z_{a,b}$ in step~(2).
ALP1-Projection defines values $y_i$ that are downscaled and rounded versions of~$x_i$, and implicitly uses a unary encoding of $y_i$.
We replace the unary encoding by a sensitivity~1 error-correcting code $\mathcal{G}=(\mathcal{G}_{dec},\mathcal{G}_{enc})$ given by our Corollary~\ref{cor:expandercode} with decoding parameter $\alpha = 1/5$.
Specifically,
\[
z_{a,b} = \begin{cases}
        1, & \exists i\in S: \mathcal{G}_{enc}(y_i)_b = 1 \text{ and } h_b(i) = a \\
        0, & \text{otherwise}
      \end{cases} \enspace .
\]
Note that this requires the number of hash functions (denoted by $m$ in~\cite{alp-mech}) to be equal to the block size $d'$ of the code $\mathcal{G}$, discussed below.
Since $\mathcal{G}_{enc}$ has sensitivity 1 and since $\mathcal{G}_{enc}(0)=0$\footnote{The property that $\mathcal{G}_{enc}(0)=0$ holds for the particular code in Corollary~\ref{cor:expandercode} but can be achieved in general by permuting the codewords using $x\mapsto x \oplus \mathcal{G}_{enc}(0)$, changing no other properties.} we see that adding an element to~$S$ changes at most one value $z_{a,b}$, and the same holds for removing an element.
This means that the privacy of the modified ALP1-Projection algorithm, which applies randomized response to each bit $z_{a,b}$, follows exactly as in~\cite{alp-mech}.

The ALP1-Estimator algorithm (Algorithm 3 in~\cite{alp-mech}), tailored to decoding the noisy unary encoding, must be replaced by running $\mathcal{G}_{dec}$ on the relevant bits of the data structure.
For $i\in [u]$ if we let $v^{(i)}\in \{0,1\}^{d'}$ be the vector given by $v^{(i)}_b = \tilde{z}_{h_b(i),b}$, on input $i$ the new ALP1-Estimator returns $\mathcal{G}_{dec}(v^{(i)})$.
We can write $v^{(i)} = \mathcal{G}_{enc}(v) \oplus \eta^{(i)}$ for a vector $\eta^{(i)}$ where $\eta^{(i)}_b$ represents errors due to a hash collision $h_b(i)=h_b(i')$ for some $i'\in S\backslash \{i\}$ or due to a bit flip introduced by randomized response on the variable $z_{h_b(i),b}$.
Since the hash functions and noise bits are independent, the bits of $\eta^{(i)}$ are independent.
Corollary~\ref{cor:expandercode} is stated for noise distribution $\Bern(\alpha/2)^{d'}$, but since expander codes enjoy a worst-case guarantee on decoding radius it is easy to see that the result holds as long as there is an \emph{upper bound} of $\alpha/2$ on the probability of flipping each bit.
Thus, it suffices to argue that for each $b$, $\Pr[\eta^{(i)}_b = 1] \leq 1/10$.
This can be achieved by a union bound on two events: 1) bit $b$ is flipped by randomized response, and 2) there is a hash collision $h_b(i)=h_b(i')$.
Choosing the hash table size $s$ and the randomized response parameter such that the probability for each of these events is bounded by $1/20$ changes only constant factors in the error bound and space usage.
Here we make use of the fact that privacy for ALP1-Projection relies on a combination of scaling and randomized response so that we can choose the parameter of randomized response to be any constant without affecting privacy.

The ALP-Projection/Estimation algorithms (Algorithms 4 and 5 in~\cite{alp-mech}) that scale the inputs to achieve sensitivity $\varepsilon$ work unchanged, except for relying on the changed ALP1-Projection/Estimation algorithms.
To finalize the argument, the Threshold ALP-Projection/Estimation algorithms (Algorithms 7 and 8 in~\cite{alp-mech}) are changed to use the new versions of the ALP1 Projection/Estimator algorithms.
Since it suffices to encode numbers up to $\ell=O(\log u)$, we can choose the code with message length $\log_2 \ell = O(\log\log u)$, and hence also with block length $d'=O(\log\log u)$.
Since decoding takes expected $O(d')$ time, estimating an item frequency can be done in expected time $O(\log\log u)$.

Finally, to bound the estimation error we consider the two cases in Threshold ALP-Estimation. 
If the output is generated by the threshold Laplace mechanism the expected and high-probability bounds follow from the standard analysis of the Laplace mechanism.
Otherwise, if the output is generated by ALP-Estimation we note that the final error is the error $|\mathcal{G}_{dec}(v^{(i)}) - y_i|$ of the error-correcting code multiplied by $O(1/\varepsilon)$.
By Corollary~\ref{cor:expandercode}, using the bound on $\eta^{(i)}$ above, this error is $O(1)$ in expectation and is bounded by $\ell=O(\log(u))$ with probability~1 since we always decode to $[\ell]$.
This finishes the error analysis.

\end{appendices}
\end{document}